\documentclass[journal]{IEEEtran}

\usepackage{amsmath,amssymb,mathtools,amsthm}
\usepackage{booktabs,tabularx,array}
\usepackage{tikz}
\usetikzlibrary{positioning,calc}
\usepackage{cite}
\IfFileExists{xurl.sty}{\usepackage{xurl}}{\usepackage[hyphens]{url}}
\usepackage[hidelinks]{hyperref}
\allowdisplaybreaks

\newtheorem{theorem}{Theorem}
\newtheorem{lemma}[theorem]{Lemma}
\newtheorem{proposition}[theorem]{Proposition}
\newtheorem{corollary}[theorem]{Corollary}

\newcommand{\AEnt}[1][n]{\overline{\Gamma^*_{#1}}}

\newcommand{\MixL}{\widetilde\Gamma^{\mathrm{MixL}}}
\newcommand{\Abl}{\widetilde\Gamma^{\mathrm{Abl}}}
\newcommand{\Hom}{\widetilde\Gamma^{\mathrm{Hom}}}

\newcommand{\Ftwo}{\mathbb F_2}

\newcommand{\conv}{\operatorname{conv}}

\hypersetup{
  pdftitle={Separating Abelian and Homomorphic Entropy Cones},
  pdfauthor={Shahram Khazaei},
  pdfsubject={Entropy regions induced by Abelian and homomorphic group-characterizable random variables},
  pdfkeywords={entropy region, group-characterizable random variables, Abelian random variables, homomorphic random variables, common information, subgroup lattice}
}

\title{Separating Abelian and Homomorphic Entropy Cones}
\author{Shahram~Khazaei%
\thanks{S. Khazaei is with the Department of Mathematical Sciences, Sharif University of Technology, Tehran, Iran (e-mail: shahram.khazaei@sharif.ir).}}

\begin{document}
\maketitle

\begin{abstract}
Chan and Yeung showed that finite groups suffice to determine which homogeneous
linear information inequalities are universally valid.  We compare two
restricted group-characterizable entropy cones: the Abelian cone $\Abl_n$ and
the homomorphic cone $\Hom_n$, the latter generated by coset systems of normal
subgroups.  We prove
\[
  \Abl_{16}\subsetneq\Hom_{16},
\]
and, if $n_{\rm AH}$ is the least number of variables for which these cones
differ, we show $6\le n_{\rm AH}\le16$.  The separating functional is a
class-restricted entropy inequality: it is valid on the Abelian cone but is not
a universal information inequality.  It is obtained by lifting the order dual
of the P\'alfy--Szab\'o six-cross identity while quantifying errors at inexact
subgroup joins.  We then construct sixteen normal subgroups of a class-two
$2$-group of order $2^{43}$ for which every join error vanishes while the
endpoint containment fails by one bit.  Since mixed-linear random variables
are Abelian, the same example also separates the mixed-linear and homomorphic
entropy cones.
\end{abstract}

\begin{IEEEkeywords}
entropy region, group-characterizable random variables, Abelian random
variables, homomorphic random variables, common information, subgroup lattice.
\end{IEEEkeywords}

\section{Introduction}\label{sec:introduction}
For $n$ jointly distributed random variables $X_1,\ldots,X_n$, write
$[n]=\{1,\ldots,n\}$ and consider the entropy vector
\[
  h_X=\bigl(h_X(A):\varnothing\ne A\subseteq[n]\bigr),
  \qquad h_X(A)=H(X_A).
\]
Thus $h_X$ is a vector in $\mathbb R^{2^n-1}$.  The set of such vectors is the
$n$-variable entropy region, and its closure is the almost-entropic cone
$\AEnt$.  A linear information inequality is an inequality of the form
\[
   \sum_{\varnothing\ne A\subseteq[n]} \kappa_A h(A)\ge0.
\]
Understanding which such inequalities are valid on $\AEnt$ is a basic problem
of information theory.

Chan and Yeung gave a group-theoretic model for this problem
\cite{ChanYeung2002}.  Let $G$ be a finite group, let
$G_1,\ldots,G_n\le G$, and let $\Theta$ be uniform on $G$.  If $X_i=\Theta G_i$ is the
left-coset observation, then for $A\subseteq[n]$,
\begin{equation}\label{eq:gc-entropy-intro}
  H(X_A)=\log |G:G_A|,
  \qquad
  G_A=\bigcap_{i\in A}G_i,
\end{equation}
with $G_\varnothing=G$.  After closure and convexification, these
group-characterizable (GC) entropy vectors generate the whole almost-entropic
cone.  Consequently, a homogeneous linear information inequality is valid for
all random variables if and only if the inequality obtained by substituting
$h(A)=\log|G:G_A|$ is valid for every finite group and every choice of
subgroups.

Algebraic restrictions on the group model give subclasses of GC random
variables, and therefore nested closed entropy cones; some of the inclusions
may be equal.  Four classes are relevant here.
\begin{itemize}
\item A \emph{linear} tuple comes from subspaces of a finite vector space.
Independent products of finitely many linear tuples, whose factors may be over
different finite fields, are called \emph{mixed-linear}.  Their closed entropy
cone is denoted by $\MixL_n$ \cite{JafariKhazaei2021}.
\item An \emph{Abelian} tuple is a GC tuple whose ambient group is Abelian.  Its
closed entropy cone is denoted by $\Abl_n$.
\item A \emph{homomorphic} tuple is a GC tuple for which every coset map
$G\to G/G_i$ is a group homomorphism.  Equivalently, every $G_i$ is normal in
$G$ \cite{KaboliKhazaeiParviz2021}.  Its closed entropy cone is denoted by
$\Hom_n$.
\item Without an algebraic restriction one obtains the GC class, whose closure
and convexification give $\AEnt$ by the Chan--Yeung theorem.
\end{itemize}
In the first three cases ``entropy cone'' means the closed conic hull of the
entropy vectors of the indicated class.  These definitions give
\begin{equation}\label{eq:hierarchy-intro}
  \MixL_n\subseteq\Abl_n\subseteq\Hom_n\subseteq\AEnt.
\end{equation}
We do not know whether the first inclusion can be strict.  This paper proves
that the second can.

\begin{theorem}[Main theorem]\label{thm:main}
There exist a finite group $G$ of order $2^{43}$ and sixteen normal subgroups
$N_1,\ldots,N_{16}\triangleleft G$ such that the associated GC entropy vector
lies in $\Hom_{16}\setminus\Abl_{16}$.  Consequently,
\[
  \MixL_{16}\subseteq\Abl_{16}\subsetneq\Hom_{16}.
\]
\end{theorem}

If $n_{\rm AH}$ is the smallest number of variables for which the Abelian and
homomorphic cones differ, then
\[
  6\le n_{\rm AH}\le16.
\]
The upper bound is Theorem~\ref{thm:main}.  For the lower bound we use the
complete five-variable description of Dougherty, Freiling, and Zeger
\cite{DFZLinearRank}.  Their facet inequalities for the five-variable
linear-rank cone admit common-information proofs, while its extreme rays are
realized by linear random variables independently of the field characteristic.
The same common-information steps are available for normal-subgroup coset
variables, since intersections and products of normal subgroups remain normal
and the product represents the corresponding common information.  Hence the
mixed-linear, Abelian, and homomorphic cones coincide through five variables.

The proof sits at the intersection of linear-rank inequalities, common
information, and subgroup lattices.  Linear random variables correspond to
subspaces, so their entropy inequalities are inequalities for subspace ranks.
Ingleton's inequality is the classical example \cite{Ingleton1971}; later
families were obtained by Dougherty, Freiling, and Zeger and by Kinser
\cite{DFZLinearRank,Kinser2011}.  Common-information extensions provide proofs
of many such inequalities and are useful in secret-sharing lower bounds
\cite{MartinPadroYang2016,FarrasEtAl2020,BamiloshinEtAl2021}.  Related
extension properties and Abelian connections are studied in
\cite{BamiloshinFarrasPadro2025,MejiaMontoya2022}.  Kaboli, Khazaei, and Parviz proved that homomorphic GC tuples have a common
information for every pair of subcollections
\cite[Cor.~6.3]{KaboliKhazaeiParviz2021}.  In the normal-subgroup model the same
observation iterates: intersections remain normal, and if $H,K\triangleleft G$
then $HK\triangleleft G$ and the coset variable with kernel $HK$ is a common
information of the coset variables with kernels $H$ and $K$.  Hence every
finite sequence of common-information extensions used in a linear-rank proof
can be realized within the homomorphic class.

Field characteristic gives a different source of rank inequalities.  For at
most five variables the complete linear-rank cone is characteristic-independent.
The six-variable cone is not completely known, and it remains open whether
field characteristic is relevant there \cite{Dougherty2014Six}.  From seven
variables onward characteristic-dependent inequalities are known; the first
examples use the Fano/non-Fano characteristic split, and Dougherty, Freiling,
and Zeger later gave eight-variable inequalities with network-coding
applications \cite{BlasiakKleinbergLubetzky2011,DFZ2015Characteristic}.
Jafari and Khazaei used the same Fano/non-Fano incompatibility to separate
linear and mixed-linear secret-sharing information ratios
\cite{JafariKhazaei2021}.

More recently, B\'erczi, Geh\'er, Imolay, Lov\'asz, Padr\'o, and Schwarcz
used tensor-product extension properties to obtain a
characteristic-independent rank inequality beyond the common-information
property \cite{BercziEtAl2026}.  Their result gives a new obstruction for folded
skew-representability and illustrates that extension properties stronger than
common information can produce genuinely new rank inequalities.  The present
paper uses a different obstruction, coming from Abelian subgroup lattices, to
separate Abelian from homomorphic GC entropy vectors.

Our starting point is the six-cross identity of P\'alfy and Szab\'o for
subgroup lattices of Abelian groups
\cite{PalfySzabo1995Congruence,PalfySzabo1995,Penttila2017}.  This choice is
structural: the lattice of normal subgroups of an arbitrary group is modular,
so a consequence of modularity alone cannot distinguish the Abelian and
homomorphic classes.  The six-cross law is a genuinely stronger Abelian-lattice
phenomenon and can fail for normal subgroup lattices.  What is still needed for
entropy is robustness: an auxiliary random variable standing for a subgroup
join need not realize that join exactly.  Section~\ref{sec:lift} develops a
quantitative lifting argument that controls this error.

The homomorphic counterexample is explicit.  We construct a class-two
$2$-group and sixteen normal subgroups.  Normality makes the ordinary joins in
the lifted inequality exact, while two additional containments eliminate the
remaining join errors.  The right-hand side is then zero, whereas the endpoint
conditional entropy is one bit.

The theorem concerns entropy regions.  The corresponding operational question
in secret sharing remains open: it is not known whether, for a fixed access
structure, a homomorphic scheme can outperform every mixed-linear scheme.  We
return briefly to this point in Section~\ref{sec:discussion}.

\section{Entropy Regions Induced by Groups}\label{sec:regions}
Let $G$ be a finite group, let $G_1,\ldots,G_n\le G$, let $\Theta$ be uniform on
$G$, and put $X_i=\Theta G_i$.  The tuple $X=(X_1,\ldots,X_n)$ is called
\emph{group-characterizable}, or simply \emph{GC}.  For $A\subseteq[n]$, put
\[
  G_A=\bigcap_{i\in A}G_i,\qquad G_\varnothing=G.
\]
The joint random variable $X_A$ is equivalent to the coset observation $\Theta G_A$,
so
\begin{equation}\label{eq:gc-entropy}
  h_X(A)=H(X_A)=\log|G:G_A|.
\end{equation}

The tuple is \emph{Abelian} if $G$ is Abelian.  It is \emph{homomorphic} if
each $G_i$ is normal: then $G/G_i$ is a quotient group and the observation map
$g\mapsto gG_i$ is a homomorphism; conversely, the kernel of a group
homomorphism is normal \cite{KaboliKhazaeiParviz2021}.  It is \emph{linear} if
$G$ is the additive group of a finite vector space and each $G_i$ is a
subspace.

For a set $\mathcal S$ of entropy vectors, write
$\operatorname{ccone}(\mathcal S)=\overline{\operatorname{cone}(\mathcal S)}$
for its closed conic hull.  For a prime $p$, let $\mathcal L_n^p$ be the set of
entropy vectors of linear tuples over finite fields of characteristic $p$, and
define
\[
  \widetilde\Gamma_n^p=\operatorname{ccone}(\mathcal L_n^p).
\]
A mixed-linear tuple is an independent product of finitely many linear tuples,
possibly over different characteristics.  Since entropy vectors add under
independent products,
\begin{equation}\label{eq:mixl-precise}
 \MixL_n
 =\operatorname{ccone}\!\left(\bigcup_p\mathcal L_n^p\right)
 =\overline{\conv\!\left(\bigcup_p\widetilde\Gamma_n^p\right)}.
\end{equation}
Thus the definition explicitly allows direct products of linear systems over
different finite fields.

Let $\Abl_n$ and $\Hom_n$ be the closed conic hulls of the entropy vectors of,
respectively, Abelian and homomorphic tuples.  Since linear tuples are Abelian
and direct products of Abelian groups are Abelian,
\begin{equation}\label{eq:basic-inclusions}
  \MixL_n\subseteq\Abl_n\subseteq\Hom_n.
\end{equation}

\begin{lemma}\label{lem:closure}
Let $\Lambda(h)$ be a homogeneous linear functional of the joint entropies.  If
$\Lambda(h)\ge0$ for every Abelian tuple, then it is nonnegative on $\Abl_n$,
and hence also on $\MixL_n$.
\end{lemma}
\begin{proof}
The half-space $\{h:\Lambda(h)\ge0\}$ is a closed convex cone containing every
Abelian entropy vector, hence their closed conic hull.  The last assertion
follows from \eqref{eq:basic-inclusions}.
\end{proof}

\section{From the Six-Cross Identity to an Entropy Inequality}\label{sec:lift}
The P\'alfy--Szab\'o identity
\cite{PalfySzabo1995Congruence,PalfySzabo1995,Penttila2017} is an exact
statement in the subgroup lattice of an Abelian group.  We first recall that lattice statement, and then develop the
additional argument needed when its joins are represented by arbitrary
auxiliary coset variables.

\subsection{The Abelian subgroup lattice and the six-cross law}\label{sec:six-cross-lattice}
Let $A$ be a finite Abelian group.  Its subgroups are partially ordered by
inclusion, with
\[
   U\wedge V=U\cap V,
   \qquad
   U\vee V=U+V.
\]
We write $\cap$ and $+$ for meet and join.  For notational simplicity,
throughout this subsection we identify a subgroup $A_z$ with its index $z$.
Choose eight subgroups
\[
 x_1,\ldots,x_4,y_1,\ldots,y_4\le A.
\]
From them define the following exact lattice expressions:
\begin{equation}\label{eq:q-def}
 q_{ij}=(x_i\cap y_j)+(y_i\cap x_j),
 \quad ij\in\{12,34,13,24,23\},
\end{equation}
and
\begin{align}
 c&=(q_{12}\cap q_{34})+(q_{13}\cap q_{24}),\label{eq:c-def}\\
 d&=(c\cap q_{23})+(x_4\cap y_1),\label{eq:d-def}\\
 a_{23}&=(x_2\cap y_2)+(x_3\cap y_3),\label{eq:a23-def}\\
 a_{234}&=a_{23}+(x_4\cap y_4).\label{eq:a234-def}
\end{align}
The six-cross identity of P\'alfy and Szab\'o has an order-dual form because
finite Pontryagin duality reverses the subgroup-lattice order.  The form needed
here is
\begin{equation}\label{eq:six-cross}
  d\cap y_4\le x_1+(y_1\cap a_{234}).
\end{equation}

\begin{theorem}[P\'alfy--Szab\'o, order-dual form]\label{thm:six-cross}
Containment \eqref{eq:six-cross} holds for every choice of the eight base
subgroups in a finite Abelian group
\cite{PalfySzabo1995Congruence,PalfySzabo1995,Penttila2017}.
\end{theorem}
\begin{proof}
The result is the order dual of the P\'alfy--Szab\'o identity.  We include a
direct verification of this formulation.  Take $v\in d\cap y_4$ and
write $v=u+w$ with $u\in c\cap q_{23}$ and
$w\in x_4\cap y_1$.  Since $u\in c$, write $u=r+s$ with
$r\in q_{12}\cap q_{34}$ and
$s\in q_{13}\cap q_{24}$.  Choose decompositions
\[
 r=\alpha_{12}+\beta_{12}=\alpha_{34}+\beta_{34},
 \qquad
 s=\alpha_{13}+\beta_{13}=\alpha_{24}+\beta_{24},
\]
\[
 u=\alpha_{23}+\beta_{23},
\]
where $\alpha_{ij}\in x_i\cap y_j$ and
$\beta_{ij}\in x_j\cap y_i$.  Put
\[
 x=\alpha_{12}+\alpha_{13}\in x_1,
 \qquad
 y=\beta_{12}+\beta_{13}+w\in y_1.
\]
Then $v=x+y$.  Define
\begin{align*}
 t_2&=\beta_{12}+\alpha_{24}-\alpha_{23}
     =\beta_{23}-\alpha_{12}-\beta_{24}\in x_2\cap y_2,\\
 t_3&=\beta_{13}+\alpha_{34}-\beta_{23}
     =\alpha_{23}-\alpha_{13}-\beta_{34}\in x_3\cap y_3,\\
 t_4&=\beta_{34}+\beta_{24}+w
     =v-\alpha_{34}-\alpha_{24}\in x_4\cap y_4.
\end{align*}
The alternative expressions follow from the three decompositions of $u$.
Finally $t_2+t_3+t_4=y$, so $y\in a_{234}$ and
$v\in x_1+(y_1\cap a_{234})$.
\end{proof}

Theorem~\ref{thm:six-cross} assumes that every join is evaluated exactly.  In
an entropy inequality, however, an auxiliary random variable may only
approximate the common information corresponding to such a join.  The next
three subsections quantify and propagate this error.

\subsection{A subgroup distance}\label{sec:distance}
For subgroups $V\le U$, recall that $|U:V|=|U|/|V|$.  For arbitrary subgroups
$U,W\le A$, define
\begin{align*}
  \rho(U,W)&=\log |U:U\cap W|,\\
  \delta(U,W)&=\rho(U,W)+\rho(W,U).
\end{align*}
Thus $\rho(U,W)$ measures the logarithmic index of the part of $U$ not
contained in $W$, and $\delta$ is its symmetric version.

We shall also use the standard finite Pontryagin duality identities.  Write
$\widehat A=\operatorname{Hom}(A,\mathbb C^\times)$ and
\[
 U^\perp=\{\chi\in\widehat A:\chi(u)=1\text{ for every }u\in U\}.
\]
Then $|U^\perp|=|A|/|U|$ and
$(U\cap W)^\perp=U^\perp+W^\perp$.

\begin{lemma}[Basic properties]\label{lem:distance}
For subgroups $U,V,W,T\le A$,
\begin{align}
 \rho(U,W)&\le \rho(U,V)+\rho(V,W),\label{eq:rho-triangle}\\
 \delta(U+T,W+T)&\le \delta(U,W),\label{eq:sum-nonexp}\\
 \delta(U\cap T,W\cap T)&\le \delta(U,W).\label{eq:meet-nonexp}
\end{align}
Consequently, $\delta$ satisfies the triangle inequality.
\end{lemma}
\begin{proof}
Since $\rho(U,W)=\log|U+W:W|$, the chain
$W\le V+W\le U+V+W$ gives
\begin{align*}
 |U+W:W|
 &\le |U+V+W:V+W|\,|V+W:W|\\
 &\le |U+V:V|\,|V+W:W|,
\end{align*}
which proves \eqref{eq:rho-triangle}.  For \eqref{eq:sum-nonexp}, the quotient
map $U+W\to(U+W+T)/(W+T)$ gives
$\rho(U+T,W+T)\le\rho(U,W)$, and the reverse directed inequality follows by
exchanging $U$ and $W$.  Under Pontryagin duality,
$\rho(U,W)=\rho(W^\perp,U^\perp)$; the directed distance is reversed, so the
symmetric quantity $\delta$ is preserved.  Duality also exchanges intersection
and sum, hence \eqref{eq:meet-nonexp} follows from \eqref{eq:sum-nonexp}.
Applying \eqref{eq:rho-triangle} in both directions gives the triangle
inequality for $\delta$.
\end{proof}

\subsection{The error of an inexact common information}\label{sec:join-error}
Recall that a random variable $V$ is a common information of random variables
$X$ and $Y$ when
\[
 H(V\mid X)=H(V\mid Y)=0,
 \qquad H(V)=I(X;Y).
\]
Under the first two equalities, the last one is equivalent to
$I(X;Y\mid V)=0$.  For Abelian GC random variables with kernels $K_X,K_Y$, the
common information is the coset variable with kernel $K_X+K_Y$.

Let $V$ be any auxiliary coset variable, with kernel $K_V$, used in place of
that exact sum.  Define its local join error by
\begin{equation}\label{eq:defect}
 \Delta(V;X,Y)
 :=H(V\mid X)+H(V\mid Y)+I(X;Y\mid V).
\end{equation}
Thus $\Delta(V;X,Y)=0$ exactly when $V$ is a common information of $X$ and
$Y$; in the Abelian coset model, this is equivalent to $K_V=K_X+K_Y$.

\begin{lemma}[Local join error]\label{lem:local-defect}
Put $W=K_V$, $S=K_X+K_Y$, and
$P=(K_X\cap W)+(K_Y\cap W)$.  Then
\begin{equation}\label{eq:local-identity}
\begin{aligned}
 \Delta(V;X,Y)-\delta(W,S)
  ={}&\log|K_X\cap K_Y:K_X\cap K_Y\cap W|\\
    &+2\log|W\cap S:P|.
\end{aligned}
\end{equation}
In particular,
\[
   \delta(K_V,K_X+K_Y)\le \Delta(V;X,Y).
\]
\end{lemma}
\begin{proof}
The coset entropy formula gives
$H(V\mid X)=\log|K_X:K_X\cap W|$ and the analogous expression for $Y$.
Conditioning on $V$ leaves a uniform coset of $W$; after translation to $W$,
the two kernels are $K_X\cap W$ and $K_Y\cap W$.  Hence
$I(X;Y\mid V)=\log|W:P|$.  Substitution, together with
$|K_X+K_Y|\,|K_X\cap K_Y|=|K_X|\,|K_Y|$, gives
\eqref{eq:local-identity}.  Here $P\le W\cap S$, and also $K_X\cap K_Y\cap W\le K_X\cap K_Y$;
therefore both terms on the right are logarithms of subgroup indices and are
nonnegative.
\end{proof}

\subsection{Accumulating the local errors}\label{sec:accumulating}
Consider first the parenthesized subgroup expression
\[
  s=((U+V)\cap W)+(Y\cap Z).
\]
Its structure is represented by the expression tree in
Fig.~\ref{fig:expression-tree}.  Leaves are subgroup labels and each internal
vertex records whether its two incoming subexpressions are combined by
intersection or by sum.

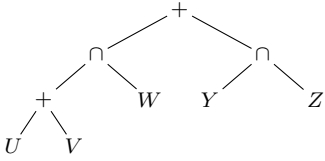
\begin{figure}[!t]
\centering
\begin{tikzpicture}[>=stealth,level distance=6mm,
  level 1/.style={sibling distance=22mm},
  level 2/.style={sibling distance=14mm},
  level 3/.style={sibling distance=8mm},
  every node/.style={inner sep=1.2pt,font=\footnotesize}]
\node {$+$}
  child { node {$\cap$}
    child { node {$+$}
      child { node {$U$} }
      child { node {$V$} } }
    child { node {$W$} } }
  child { node {$\cap$}
    child { node {$Y$} }
    child { node {$Z$} } };
\end{tikzpicture}
\caption{Expression tree of $s=((U+V)\cap W)+(Y\cap Z)$.}
\label{fig:expression-tree}
\end{figure}

Suppose the inner sum $U+V$ is represented by an auxiliary coset variable $Q$.
Write $\Delta_Q$ for the local join error \eqref{eq:defect} at this $+$-vertex.
Lemma~\ref{lem:local-defect} gives
$\delta(K_Q,U+V)\le\Delta_Q$.  The left subexpression is
$t=(U+V)\cap W$; with $Q$ in place of the exact sum, its actual kernel is
$K_Q\cap W$, whereas its exact kernel is $(U+V)\cap W$.  By
\eqref{eq:meet-nonexp},
\[
 \delta\bigl(K_Q\cap W,(U+V)\cap W\bigr)\le\Delta_Q.
\]
This example contains the entire propagation idea: each approximate sum
creates one local error, and later intersections or sums cannot amplify errors
already present in its two input subexpressions.

For a general subexpression $t$, define $K_t$ by evaluating each intersection
vertex exactly and using the chosen auxiliary kernel at each sum vertex.  Define
$\widehat K_t$ recursively from the same leaves by evaluating \emph{every}
intersection as subgroup intersection and \emph{every} sum as the actual
subgroup sum.  Finally, let $J(t)$ be the set of sum vertices in the expression
tree of $t$.  In Fig.~\ref{fig:expression-tree}, the left subexpression
$(U+V)\cap W$ has one element in $J(t)$, while the full expression $s$ has two:
the inner $U+V$ vertex and the root sum vertex.

\begin{lemma}[Propagation]\label{lem:propagation}
For every subexpression $t$,
\[
  \delta(K_t,\widehat K_t)\le\sum_{v\in J(t)}\Delta_v.
\]
\end{lemma}
\begin{proof}
Induct upward through the expression tree.  At an intersection vertex, replace
the two input kernels one at a time and use \eqref{eq:meet-nonexp}; the triangle
inequality adds their accumulated errors.  At a sum vertex, first replace the
chosen auxiliary kernel by the exact sum of its two current input kernels,
which costs at most its local error by Lemma~\ref{lem:local-defect}.  Then
replace the two inputs one at a time and use \eqref{eq:sum-nonexp}.  Summing the
errors gives exactly the right-hand side.
\end{proof}

\begin{theorem}[Lifting theorem]\label{thm:defect-lift}
Suppose a subgroup-lattice containment $r\le \ell$ holds in every finite
Abelian group.  Represent intersections by joint random variables and assign an
auxiliary random variable to every sum.  Then
\begin{equation}\label{eq:defect-lift}
 H(X_\ell\mid X_r)
 \le
 \sum_{v\in J(r)}\Delta_v+
 \sum_{v\in J(\ell)}\Delta_v.
\end{equation}
The inequality remains valid after identifying any of the formal random
variables.
\end{theorem}
\begin{proof}
For the exact evaluations, $\widehat K_r\le\widehat K_\ell$.  Hence
\[
 \rho(K_r,K_\ell)
 \le \delta(K_r,\widehat K_r)+\delta(K_\ell,\widehat K_\ell).
\]
The left side is $H(X_\ell\mid X_r)$, and
Lemma~\ref{lem:propagation} bounds the two terms on the right.  Finally,
identifying formal variables is simply a particular choice of the auxiliary
variables, while the argument above holds for arbitrary choices.
\end{proof}

\subsection{The sixteen-variable Abelian inequality}\label{sec:sixteen-abelian}
The lifted inequality uses eight base random variables
$X_{x_1},\ldots,X_{x_4},X_{y_1},\ldots,X_{y_4}$ and eight auxiliary random
variables $X_{q_{12}},X_{q_{34}},X_{q_{13}},X_{q_{24}},X_{q_{23}},X_c,X_d,
X_{a_{23}}$.  Thus the entropy vector has sixteen coordinates.  The two
remaining joins in the formal six-cross expression reuse $X_{a_{23}}$ and
$X_{x_1}$ rather than introducing two additional coordinates.

Let
\begin{multline*}
 \mathcal Z=\{x_1,x_2,x_3,x_4,y_1,y_2,y_3,y_4,\\
 q_{12},q_{34},q_{13},q_{24},q_{23},c,d,a_{23}\}.
\end{multline*}
Let $X=(X_z)_{z\in\mathcal Z}$ be an Abelian GC tuple and write $K_z$ for the
kernel of $X_z$.  For the rest of this subsection, we identify each random
variable $X_z$ with its index $z$.  The five variables $q_{ij}$ and the variables
$c,d,a_{23}$ are assigned to the corresponding formal joins in
\eqref{eq:q-def}--\eqref{eq:a23-def}; their kernels need not equal the exact
joins.  At the next formal join,
\[
  K_{a_{23}}+(K_{x_4}\cap K_{y_4}),
\]
we use $a_{23}$ again as the auxiliary.  At the final join,
\[
  K_{x_1}+(K_{y_1}\cap K_{a_{23}}),
\]
we use $x_1$ again.  The two corresponding local errors reduce to
\[
 H(a_{23}\mid x_4,y_4),
 \qquad
 H(x_1\mid y_1,a_{23}).
\]
For random variables $V,P,Q$, retain the notation
\[
 \Delta(V;P,Q)=H(V\mid P)+H(V\mid Q)+I(P;Q\mid V).
\]
Juxtaposition such as $x_i y_j$ denotes the corresponding joint random
variable.  Put $\mathcal I=\{12,34,13,24,23\}$.

\begin{theorem}[Sixteen-variable Abelian entropy inequality]\label{thm:abelian-ineq}
Every Abelian GC tuple $X=(X_z)_{z\in\mathcal Z}$ satisfies
\begin{align}\label{eq:compact-ineq}
 H(x_1\mid d,y_4)
 \le{}&
 \sum_{ij\in\mathcal I}
 \Delta(q_{ij};x_i y_j,y_i x_j) \\
 &+\Delta(c;q_{12}q_{34},q_{13}q_{24})\nonumber\\
 &+\Delta(d;cq_{23},x_4y_1)\nonumber\\
 &+\Delta(a_{23};x_2y_2,x_3y_3)\nonumber\\
 &+H(a_{23}\mid x_4,y_4)\nonumber\\
 &+H(x_1\mid y_1,a_{23}).\nonumber
\end{align}
Its expansion is a homogeneous linear entropy inequality with $54$ nonzero
integer coefficients, listed in Appendix~\ref{app:coefficients}.
\end{theorem}
\begin{proof}
Apply Theorem~\ref{thm:defect-lift} to the P\'alfy--Szab\'o containment
\eqref{eq:six-cross}, assigning the sixteen random variables as described
above.  The two reused auxiliaries give the last two conditional-entropy terms.
\end{proof}

\begin{corollary}\label{cor:abelian-cone}
Inequality \eqref{eq:compact-ineq} is valid on $\Abl_{16}$ and hence on
$\MixL_{16}$.
\end{corollary}
\begin{proof}
Each conditional entropy and conditional mutual information in
\eqref{eq:compact-ineq} is a linear combination of joint entropies, so the
expanded form is homogeneous and linear.  Apply Lemma~\ref{lem:closure}.
\end{proof}

\section{A homomorphic vector violating the inequality}\label{sec:witness}
\subsection{The group}
Let $V=\Ftwo^8$.  Let $C$ be the $36$-dimensional $\Ftwo$-space with basis
\[
  s_1,\ldots,s_8,\qquad c_{ij}\quad(1\le i<j\le8).
\]
Define the bilinear map $f:V\times V\to C$ by
\begin{equation}\label{eq:cocycle}
 f(u,v)=\sum_{i=1}^8u_iv_i s_i
       +\sum_{1\le i<j\le8}u_jv_i c_{ij},
\end{equation}
and put
\begin{equation}\label{eq:group-law}
 (u,a)(v,b)=(u+v,a+b+f(u,v)).
\end{equation}
This gives a class-two group $\widehat G=V\times C$.  If
$g_i=(e_i,0)$, then
\[
  g_i^2=(0,s_i),\qquad [g_i,g_j]=(0,c_{ij}).
\]
Finally set
\begin{equation}\label{eq:group-quotient}
  G=\widehat G/\langle g_6^2\rangle.
\end{equation}
Thus $|G|=2^{43}$.

Here and below, if $S\subseteq G$, then $\langle S\rangle$ denotes the normal
closure of $S$ in $G$; angle brackets in a vector
space retain their usual meaning of linear span.

\subsection{The sixteen normal subgroups}\label{sec:sixteen-subgroups}
The sixteen subgroup labels used in the entropy vector are
\[
 x_1,x_2,x_3,x_4,y_1,y_2,y_3,y_4,
 q_{12},q_{34},q_{13},q_{24},q_{23},c,d,a_{23}.
\]
We now realize the sixteen labels of Section~\ref{sec:sixteen-abelian} by
normal subgroups of $G$.  The construction is organized around the same three
features as the inequality: the ordinary joins should be exact, the two reused
joins should have zero error, and the endpoint containment should fail.  Normal
subgroups handle the first requirement automatically.  The relation graph
below is arranged so that shared edge differences create the cross
intersections needed for the other two requirements.

If an edge carries the two labels $x_i,y_j$, its edge difference lies in both
preliminary subgroups $B_{x_i}$ and $B_{y_j}$, hence in
$B_{x_i}\cap B_{y_j}$.  These intersections are the building blocks of the
$q_{ij}$ terms in \eqref{eq:q-def}.

Choose the following nine elements of $G$:
\begin{equation}\label{eq:vertices}
\begin{aligned}
 v_0&=1,&v_1&=g_2g_4g_6,&v_2&=g_2g_3,&v_3&=g_6,\\
 v_4&=g_5,&v_5&=g_4,&v_6&=g_3,&v_7&=g_2,&v_8&=g_1.
\end{aligned}
\end{equation}
For an edge joining vertices $s$ and $t$, write
\[
   [st]:=v_s^{-1}v_t.
\]
If an edge carries the label $z$, the element $[st]$ is included among the
normal generators of $B_z$.  A shared edge therefore contributes the same
group element to two preliminary subgroups.  For example, the two edges
labelled $x_1$ are $04$ and $36$, so
\[
   B_{x_1}=\langle[04],[36]\rangle,
\]
and similarly,
\[
   B_{x_2}=\langle[43],[37],[08]\rangle.
\]
Figure~\ref{fig:relation-graph} displays the incidences, and
Table~\ref{tab:preliminary-subgroups} gives all eight definitions explicitly.

\begin{figure}[t]
\centering
\begin{tikzpicture}[
  scale=.93,
  vertex/.style={circle,draw,thick,fill=white,minimum size=6mm,inner sep=0pt,font=\small},
  elab/.style={font=\scriptsize,fill=white,inner sep=1.4pt,text=black},
  every path/.style={draw=black!65,line width=.55pt}
]
\node[vertex] (v8) at (0,3.4) {8};
\node[vertex] (v4) at (-2.6,1.8) {4};
\node[vertex] (v6) at (2.6,1.8) {6};
\node[vertex] (v0) at (-4.1,0) {0};
\node[vertex] (v3) at (0,0) {3};
\node[vertex] (v2) at (4.1,0) {2};
\node[vertex] (v5) at (-2.6,-1.8) {5};
\node[vertex] (v7) at (2.6,-1.8) {7};
\node[vertex] (v1) at (0,-3.4) {1};

\draw (v0)--node[elab,sloped,above] {$x_1,y_2$} (v4);
\draw (v4)--node[elab,sloped,above] {$x_2,y_1$} (v3);
\draw (v3)--node[elab,sloped,above] {$x_1,y_3$} (v6);
\draw (v6)--node[elab,sloped,above] {$x_3,y_1$} (v2);
\draw (v0)--node[elab,sloped,above,pos=.54] {$x_2,y_3$} (v8);
\draw (v8)--node[elab,sloped,above,pos=.54] {$x_3,y_2$} (v2);
\draw (v0)--node[elab,sloped,below] {$x_3,y_4$} (v5);
\draw (v5)--node[elab,sloped,below] {$x_4,y_3$} (v3);
\draw (v3)--node[elab,sloped,below] {$x_2,y_4$} (v7);
\draw (v7)--node[elab,sloped,below] {$x_4,y_2$} (v2);
\draw (v2)--node[elab,sloped,below,pos=.55] {$x_4,y_1$} (v1);
\draw (v0)--node[elab,sloped,below,pos=.55] {$y_4$} (v1);
\end{tikzpicture}
\caption{The labelled relation graph defining the eight preliminary normal
subgroups.  If the edge $st$ carries the label $z$, then $[st]=v_s^{-1}v_t$
is one of the normal generators of $B_z$.  Shared edges contribute to both
labels.}
\label{fig:relation-graph}
\end{figure}
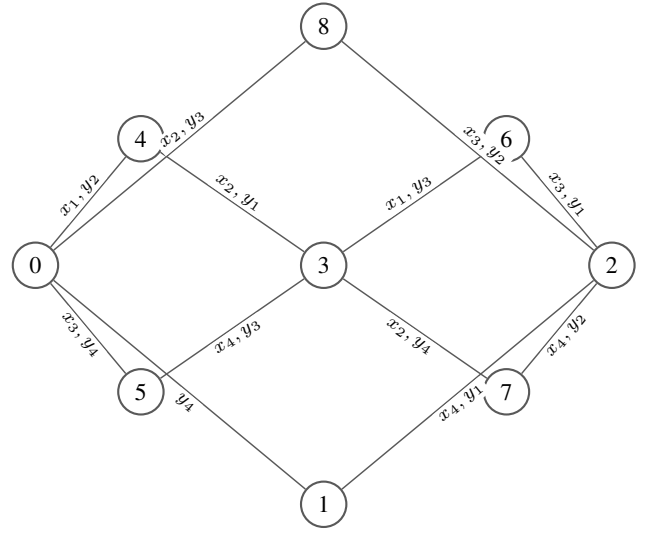

\begin{table*}[!t]
\centering
\caption{The eight preliminary normal subgroups read from
Fig.~\ref{fig:relation-graph}.  The third column expands each edge difference
using \eqref{eq:vertices}.}
\label{tab:preliminary-subgroups}
\small
\setlength{\tabcolsep}{4pt}
\begin{tabular}{c l l l}
\toprule
$z$ & edge differences & corresponding group elements & definition of $B_z$\\
\midrule
$x_1$ & $[04],[36]$ & $g_5,\ g_6^{-1}g_3$ & $\langle[04],[36]\rangle$\\
$x_2$ & $[43],[37],[08]$ & $g_5^{-1}g_6,\ g_6^{-1}g_2,\ g_1$ & $\langle[43],[37],[08]\rangle$\\
$x_3$ & $[05],[62],[82]$ & $g_4,\ g_3^{-1}g_2g_3,\ g_1^{-1}g_2g_3$ & $\langle[05],[62],[82]\rangle$\\
$x_4$ & $[53],[72],[21]$ & $g_4^{-1}g_6,\ g_3,\ g_3^{-1}g_4g_6$ & $\langle[53],[72],[21]\rangle$\\
$y_1$ & $[43],[62],[21]$ & $g_5^{-1}g_6,\ g_3^{-1}g_2g_3,\ g_3^{-1}g_4g_6$ & $\langle[43],[62],[21]\rangle$\\
$y_2$ & $[04],[72],[82]$ & $g_5,\ g_3,\ g_1^{-1}g_2g_3$ & $\langle[04],[72],[82]\rangle$\\
$y_3$ & $[53],[36],[08]$ & $g_4^{-1}g_6,\ g_6^{-1}g_3,\ g_1$ & $\langle[53],[36],[08]\rangle$\\
$y_4$ & $[05],[37],[01]$ & $g_4,\ g_6^{-1}g_2,\ g_2g_4g_6$ & $\langle[05],[37],[01]\rangle$\\
\bottomrule
\end{tabular}
\end{table*}

We now pass from the eight preliminary subgroups to the sixteen subgroups
used in the entropy vector.  First realize $a_{23}$ exactly as in
\eqref{eq:a23-def}:
\begin{equation}\label{eq:a23-witness}
 N_{a_{23}}=(B_{x_2}\cap B_{y_2})(B_{x_3}\cap B_{y_3}).
\end{equation}
The crucial first containment is
\begin{equation}\label{eq:collision}
 B_{x_4}\cap B_{y_4}\le N_{a_{23}}.
\end{equation}
It says precisely that this same subgroup $N_{a_{23}}$ can also serve for the
formal join $a_{23}+(x_4\cap y_4)$ from Section~\ref{sec:lift}.

Next set
\begin{equation}\label{eq:x1-absorb}
 N_{x_1}=B_{x_1}(B_{y_1}\cap N_{a_{23}}),
\end{equation}
and put $N_z=B_z$ for
$z\in\{x_2,x_3,x_4,y_1,y_2,y_3,y_4\}$.  By construction,
$B_{y_1}\cap N_{a_{23}}\le N_{x_1}$, so $N_{x_1}$ can serve for the final
right-hand join in the same way.

For the remaining seven labels $q_{ij},c,d$, use exactly the formulas
\eqref{eq:q-def}--\eqref{eq:d-def}, now reading $\cap$ as subgroup intersection
and $+$ as subgroup product.  For example,
\[
 N_{q_{12}}=(N_{x_1}\cap N_{y_2})(N_{y_1}\cap N_{x_2}).
\]
Together with the eight base subgroups and $N_{a_{23}}$, these are the sixteen
normal subgroups of the witness.  Normality is immediate: each $B_z$ is a
normal closure, and intersections and products of normal subgroups are normal.
For normal subgroups the product is their subgroup-lattice join, so all the
ordinary joins of Section~\ref{sec:sixteen-abelian} are realized exactly.

The only finite facts about these subgroups that are not immediate from their
definitions are summarized next.  Appendix~\ref{app:finite-certificate} gives
a direct binary-linear reduction of these facts from the displayed group law
and subgroup generators.

\begin{proposition}[Finite certificate]\label{prop:certificate}
Let
\[
 R=N_d\cap N_{y_4},\qquad L=N_{x_1},\qquad e=g_2g_4g_6.
\]
Then the containment \eqref{eq:collision} holds,
$e\in R\setminus L$, and
\[
   |R:R\cap L|=2.
\]
\end{proposition}

\begin{proof}
Appendix~\ref{app:finite-certificate} gives a binary-coordinate reduction rule
for normal closures, products, and intersections, and then supplies direct
certificates for the containment, endpoint membership, and index-two claim.
\end{proof}

\subsection{The one-bit violation}
Let $U_G$ be uniform on $G$ and, in the order
\[
 x_1,x_2,x_3,x_4,y_1,y_2,y_3,y_4,
 q_{12},q_{34},q_{13},q_{24},q_{23},c,d,a_{23},
\]
let $X_z=U_GN_z$.  Then
\begin{equation}\label{eq:witness-entropy}
 H(X_A)=43-\log_2\left|\bigcap_{z\in A}N_z\right|.
\end{equation}

At each ordinary join just described, the named subgroup is exactly the
product of its two normal child subgroups.  Therefore every
$\Delta$ term in \eqref{eq:compact-ineq} is zero.  The containment
\eqref{eq:collision} gives
\[
 H(X_{a_{23}}\mid X_{x_4},X_{y_4})=0,
\]
and \eqref{eq:x1-absorb} gives
\[
 H(X_{x_1}\mid X_{y_1},X_{a_{23}})=0.
\]
On the other hand, Proposition~\ref{prop:certificate} gives
\[
 H(X_{x_1}\mid X_d,X_{y_4})=\log_2|R:R\cap L|=1.
\]
Thus the left side of \eqref{eq:compact-ineq} is one and its right side is
zero.  This proves Theorem~\ref{thm:main}.

\section{Conclusion}\label{sec:discussion}
We have proved
\[
   \Abl_{16}\subsetneq\Hom_{16},
\]
and hence also separated the mixed-linear and homomorphic entropy cones.
Together with the five-variable DFZ consequence recalled in the Introduction,
Theorem~\ref{thm:main} gives
\[
   6\le n_{\rm AH}\le16.
\]
The present six-cross specialization uses eight base labels and eight auxiliary
labels; the two remaining formal joins are represented by already existing
labels.  Reducing the number of variables further would require additional
identifications that preserve both the Abelian lifting argument and the
zero-error homomorphic witness.  We do not know whether this can be done.  It
also remains open whether $\MixL_n\subseteq\Abl_n$ is strict for some $n$.

The witness has a direct common-information consequence.  Intersections and
products of normal subgroups are normal, and the product of two normal
subgroups represents a common information of the corresponding coset
variables.  Thus every finite sequence of common-information extensions can be
realized within the homomorphic class.  Since the witness violates
\eqref{eq:compact-ineq}, that inequality cannot be obtained by a finite proof
consisting only of common-information extensions, Shannon inequalities, and
projection.

There is also a secret-sharing interpretation.  Appendix~\ref{app:finite-certificate}
shows that the sixteen participant kernels have trivial intersection.  Choose
a maximal subgroup $N_0$ of the finite $2$-group $G$.  Then $N_0$ is normal
and $|G:N_0|=2$.  Adjoining the coset variable $\Theta N_0$ as dealer gives a
one-bit perfect homomorphic secret-sharing scheme: for every coalition kernel
$N_A$, normality implies either $N_A\le N_0$ or $N_0N_A=G$, giving perfect
reconstruction or perfect privacy, respectively.  Its participant marginal is
the entropy vector in $\Hom_{16}\setminus\Abl_{16}$.  Since coordinate
projection maps $\Abl_{17}$ into $\Abl_{16}$, the full dealer-and-share vector
also lies outside the Abelian cone.

This does not yet give an operational advantage of homomorphic schemes over
mixed-linear schemes for a fixed access structure.  Jafari and Khazaei showed
that mixed-linear schemes can outperform single-field linear schemes, and in
the ideal case every ideal homomorphic access structure has an ideal
multi-linear realization in their terminology \cite{JafariKhazaei2021}.  It
remains open whether every homomorphic scheme can be replaced, for the same
access structure, by a mixed-linear scheme with no larger contribution vector.
The analogous question for fixed network-coding instances is also open.

\appendices
\section{A Human-Readable Finite Certificate}\label{app:finite-certificate}
This appendix closes the finite group calculation in ordinary binary-linear
algebra.  No claim in the paper depends on the accompanying software.

After quotienting by $g_6^2$, write the central coordinate space as
\begin{align*}
 C'&=\langle s_1,s_2,s_3,s_4,s_5,s_7,s_8,
          c_{12},c_{13},\ldots,c_{78}\rangle_{\Ftwo},\\
 \dim C'&=35.
\end{align*}
Every element of $G$ has a unique coordinate $(u,w)\in V\times C'$.  Let
$\pi(u,w)=u$.  For a subgroup $M\le G$, put $C_M=M\cap C'$.  Every nonempty
fiber of $\pi|_M$ is a coset of $C_M$, and hence
\begin{equation}\label{eq:order-from-projection}
  \log_2|M|=\dim\pi(M)+\dim C_M.
\end{equation}

We first record the exact reduction rule used below.  Suppose
$M=\langle x_1,\ldots,x_m\rangle$ is a normal closure.  Perform Gaussian
elimination on the vectors $\pi(x_i)\in V$, carrying the corresponding group
elements through the same row operations.  The nonzero pivot rows give lifts
$r_1,\ldots,r_t\in M$ of a basis of $\pi(M)$.  Whenever a row reduces to zero
in $V$, the corresponding group product is central; record this central
\emph{dependency residual}.  Since $G$ has class two,
\begin{equation}\label{eq:normal-closure-recipe}
\begin{split}
 C_M=\big\langle{}&\text{dependency residuals},\ r_i^2,\ [r_i,g_k]:\\
 &1\le i\le t,\ 1\le k\le8\big\rangle_{\Ftwo}.
\end{split}
\end{equation}
Thus both $\pi(M)$ and $C_M$ are obtained by ordinary row reduction over
$\Ftwo$.  A product of normal subgroups is handled by applying the same rule to
the union of their generators.  For an intersection $M\cap N$, choose one lift
$m(u)\in M$ and $n(u)\in N$ for each
$u\in\pi(M)\cap\pi(N)$.  Then
\begin{equation}\label{eq:intersection-recipe}
 C_{M\cap N}=C_M\cap C_N,
\end{equation}
and $u$ occurs in $\pi(M\cap N)$ exactly when
$m(u)n(u)^{-1}\in C_M+C_N$.  Equations
\eqref{eq:normal-closure-recipe}--\eqref{eq:intersection-recipe}, together with
Table~\ref{tab:preliminary-subgroups}, are a complete finite recipe for every
subgroup calculation below.

\subsection*{A. The reused join $a_{23}+(x_4\cap y_4)$}
From Table~\ref{tab:preliminary-subgroups},
\[
 \pi(B_{x_4})=\langle e_3,e_4+e_6\rangle,
 \qquad
 \pi(B_{y_4})=\langle e_4,e_2+e_6\rangle.
\]
These two subspaces meet trivially, so $B_{x_4}\cap B_{y_4}$ is central.
Applying \eqref{eq:normal-closure-recipe} to the three generators of each
preliminary subgroup gives the following row-echelon bases of their central
parts:
\[
\begin{array}{ll}
C_{B_{x_4}}:&
 c_{48}+c_{68},\ c_{47}+c_{67},\ c_{45}+c_{56},\ c_{46},\ c_{38},\\
&c_{37},\ c_{36},\ c_{35},\ c_{34},\ c_{24}+c_{26},\ c_{23},\\
&c_{14}+c_{16},\ c_{13},\ s_4,\ s_3,\\[1mm]
C_{B_{y_4}}:&
 c_{28}+c_{68},\ c_{27}+c_{67},\ c_{25}+c_{56},\ c_{48},\ c_{47},\\
&c_{46},\ c_{45},\ c_{23}+c_{36},\ c_{34},\ c_{26},\ c_{24},\\
&c_{12}+c_{16},\ c_{14},\ s_4,\ s_2.
\end{array}
\]
Intersecting these two explicit subspaces gives
\begin{equation}\label{eq:collision-basis}
 B_{x_4}\cap B_{y_4}
 =\langle c_{46},\ c_{23}+c_{36},\ c_{34},\ c_{24}+c_{26},\ s_4\rangle.
\end{equation}
On the other hand, the reduction for $N_{a_{23}}$ contains the central
vectors
\begin{align*}
 &c_{46},\quad c_{16}+c_{36},\quad c_{16}+c_{23},\quad c_{34},\\
 &c_{16}+c_{26},\quad c_{16}+c_{24},\quad s_4.
\end{align*}
Adding the second and third vectors, and the fifth and sixth vectors, produces
the five generators in \eqref{eq:collision-basis}.  Hence
$B_{x_4}\cap B_{y_4}\le N_{a_{23}}$, proving
\eqref{eq:collision}.

\subsection*{B. Direct membership of the endpoint element}
The relation graph gives $e\in R$ without row reduction.  Indeed,
$[08]=g_1\in B_{x_2}\cap B_{y_3}$ and
$[82]=g_1^{-1}g_2g_3\in B_{x_3}\cap B_{y_2}$, so
$g_2g_3\in N_{q_{23}}$.  Also
\[
 [04]=g_5,\quad [43]=g_5^{-1}g_6
 \quad\Longrightarrow\quad g_6\in N_{q_{12}},
\]
and
\[
 [05]=g_4,\quad [53]=g_4^{-1}g_6
 \quad\Longrightarrow\quad g_6\in N_{q_{34}}.
\]
Similarly,
\[
 [36]=g_6^{-1}g_3,\quad [62]=g_3^{-1}g_2g_3
 \quad\Longrightarrow\quad g_6^{-1}g_2g_3\in N_{q_{13}},
\]
and
\[
 [37]=g_6^{-1}g_2,\quad [72]=g_3
 \quad\Longrightarrow\quad g_6^{-1}g_2g_3\in N_{q_{24}}.
\]
Therefore $g_2g_3=g_6(g_6^{-1}g_2g_3)$ lies in $N_c$ and in
$N_{q_{23}}$.  Finally,
$[21]=g_3^{-1}g_4g_6\in N_{x_4}\cap N_{y_1}$, whence
\[
 (g_2g_3)(g_3^{-1}g_4g_6)=g_2g_4g_6=e\in N_d.
\]
Since $e=[01]\in B_{y_4}=N_{y_4}$, we obtain $e\in R$.

\subsection*{C. The endpoint index and exclusion from $L$}
The projection part is small.  The reduction rules give
\[
 \pi(N_d)=\langle e_2+e_4+e_6,\ e_2+e_3\rangle,
 \qquad
 \pi(N_{y_4})=\langle e_2+e_6,\ e_4\rangle.
\]
Their intersection is
$\langle e_2+e_4+e_6\rangle=\langle\pi(e)\rangle$.  Since $e\in R$,
\begin{equation}\label{eq:Rprojection}
 \pi(R)=\langle\pi(e)\rangle.
\end{equation}
The same explicit row reduction gives the following basis of $C_R$:
\begin{align}
C_R=\langle{}&c_{28}+c_{48}+c_{68},\ c_{27}+c_{47}+c_{67},\ c_{25}+c_{56},
\nonumber\\
&s_2+s_4+c_{46},\ c_{45},\ s_2+s_4+c_{23}+c_{36},
\nonumber\\
&s_2+s_4+c_{34},\ s_2+s_4+c_{26},\ s_2+s_4+c_{24},
\nonumber\\
&s_2+s_4+c_{12}+c_{16},\ s_2+s_4+c_{14}\rangle.\label{eq:CRbasis}
\end{align}
For $L=N_{x_1}$ the projection reduction gives
\begin{equation}\label{eq:Lprojection}
 \pi(L)=\langle e_3+e_6,\ e_5,\ e_2+e_3+e_4\rangle,
\end{equation}
with respective lifts $g_3g_6$, $g_5$, and $g_2g_3g_4$ in $L$.  A row
reduction of the central part gives
\begin{align*}
C_L=\langle{}&c_{38}+c_{68},\ c_{37}+c_{67},\ c_{58},\ c_{57},\ c_{56},\\
&c_{28}+c_{38}+c_{48},\ c_{27}+c_{37}+c_{47},\ s_2+s_4+c_{46},\\
&c_{45},\ c_{36},\ c_{35},\ s_2+s_4+c_{34},\ s_2+s_4+c_{26},\ c_{25},\\
&s_2+s_4+c_{24},\ s_2+s_4+c_{23},\ s_2+s_4+c_{16},\ c_{15},\\
&s_2+s_4+c_{14},\ s_2+s_4+c_{13},\ c_{12},\ s_5,\ s_3\rangle.
\end{align*}
The basis \eqref{eq:CRbasis} lies in this span; for example,
$c_{28}+c_{48}+c_{68}$ is the sum of
$c_{28}+c_{38}+c_{48}$ and $c_{38}+c_{68}$.  Thus
\begin{equation}\label{eq:CRinL}
 C_R\le C_L.
\end{equation}

For completeness, we give a dual certificate for the load-bearing exclusion
$e\notin L$ that does not rely on the completeness of the displayed basis of
$C_L$.  Define
\begin{align*}
\lambda={}&s_2^*+c_{13}^*+c_{14}^*+c_{16}^*+c_{23}^*+c_{24}^*
          +c_{26}^*+c_{34}^*+c_{46}^*,\\
\mu={}&s_1^*+s_3^*+c_{13}^*+c_{14}^*+c_{34}^*+c_{35}^*
        +c_{36}^*+c_{45}^*+c_{46}^*,\\
\nu={}&s_1^*+s_2^*+s_3^*+c_{16}^*+c_{23}^*+c_{24}^*+c_{26}^*
        +c_{35}^*+c_{36}^*+c_{45}^*.
\end{align*}
Then $\lambda=\mu+\nu$.  In addition, over $\Ftwo$,
\begin{align*}
\nu={}&\alpha_{22}+\beta_{22},\\
\alpha_{22}={}&s_3^*+c_{23}^*+c_{24}^*+c_{35}^*+c_{36}^*+c_{45}^*+c_{46}^*,\\
\beta_{22}={}&s_1^*+s_2^*+c_{16}^*+c_{26}^*+c_{46}^*,\\[1mm]
\nu={}&\alpha_{33}+\beta_{33},\\
\alpha_{33}={}&s_1^*+s_3^*+c_{16}^*+c_{36}^*+c_{56}^*,\\
\beta_{33}={}&s_2^*+c_{23}^*+c_{24}^*+c_{26}^*+c_{35}^*+c_{45}^*+c_{56}^*.
\end{align*}
The normal-closure rule \eqref{eq:normal-closure-recipe}, applied directly to
Table~\ref{tab:preliminary-subgroups}, gives
\begin{align*}
&\lambda(C_{B_{x_1}})=0,\qquad \mu(C_{B_{y_1}})=0,\\
&\alpha_{22}(C_{B_{x_2}})=\beta_{22}(C_{B_{y_2}})=0,\\
&\alpha_{33}(C_{B_{x_3}})=\beta_{33}(C_{B_{y_3}})=0.
\end{align*}
Each statement is checked only on the squares, commutators, and dependency
residuals of at most three edge generators.

Put
$I_{22}=B_{x_2}\cap B_{y_2}$,
$I_{33}=B_{x_3}\cap B_{y_3}$, and
$J=B_{y_1}\cap N_{a_{23}}$.  The projection spaces of $I_{22}$ and $I_{33}$
are respectively
$\langle e_1+e_2+e_5\rangle$ and
$\langle e_1+e_3+e_4\rangle$, so they meet trivially.  Hence
$C_{N_{a_{23}}}=C_{I_{22}}+C_{I_{33}}$.  The two decompositions of $\nu$
above imply that $\nu$ vanishes on both $C_{I_{22}}$ and $C_{I_{33}}$, and
therefore on $C_{N_{a_{23}}}$.  Since $C_J=C_{B_{y_1}}\cap C_{N_{a_{23}}}$,
$\lambda=\mu+\nu$ vanishes on $C_J$.
Moreover,
\[
 \pi(B_{x_1})=\langle e_3+e_6,e_5\rangle,
 \qquad
 \pi(J)=\langle e_2+e_3+e_4+e_5\rangle
\]
meet trivially.  Thus
$C_L=C_{B_{x_1}}+C_J$, and $\lambda$ annihilates $C_L$.

Now use the two lifts in \eqref{eq:Lprojection} and put
\[
 t=(g_2g_3g_4)(g_3g_6)\in L.
\]
It has the same projection as $e=g_2g_4g_6$, and the group law gives
\[
 e^{-1}t=s_3+c_{34}\in C'.
\]
Since $\lambda(s_3+c_{34})=1$ whereas $\lambda(C_L)=0$, we have $e\notin L$.
Together with \eqref{eq:Rprojection} and \eqref{eq:CRinL}, every element of
$R$ is either in $C_R$ or in $eC_R$, and the second coset is disjoint from
$L$.  Therefore
\[
 R\cap L=C_R,
 \qquad
 |R:R\cap L|=2.
\]
Finally, the participant intersection used in the secret-sharing remark is
also a short finite check.  Successive intersections of the four unchanged
base subgroups $N_{y_i}=B_{y_i}$ give central parts
\begin{align*}
C_{N_{y_1}\cap N_{y_2}\cap N_{y_3}}=\langle{}&
 c_{34}+c_{35}+c_{56},\ c_{35}+c_{45},\\
&c_{34}+c_{36},\ c_{15}+c_{16}+c_{23}+c_{26},\\
&c_{13}+c_{14}+c_{15}+c_{23}+c_{24},\\
&c_{12},\ s_6\rangle.
\end{align*}
and
\[
 N_{y_1}\cap N_{y_2}\cap N_{y_3}\cap N_{y_4}=\langle s_6\rangle
\]
in $\widehat G$.  Since $G=\widehat G/\langle s_6\rangle$, the four
subgroups $N_{y_1},\ldots,N_{y_4}$ already have trivial intersection in $G$.
Consequently,
\[
  \bigcap_{z\in\mathcal Z}N_z=1.
\]

This completes the proof of Proposition~\ref{prop:certificate} and the finite
claim used in Section~\ref{sec:discussion}.

\section{Expanded Entropy Functional}\label{app:coefficients}
Write the expanded form of \eqref{eq:compact-ineq} as
\[
  \sum_{\varnothing\ne S\subseteq\mathcal Z}\kappa_S H(X_S)\ge0,
\]
with the convention ``right-hand side minus left-hand side.''  Table~\ref{tab:coefficients}
lists the $54$ nonzero coefficients; every omitted coefficient is zero.  In
the table, $H(z_1,\ldots,z_k)$ abbreviates
$H(X_{z_1},\ldots,X_{z_k})$.  The coefficient histogram is $35$ entries equal
to $-1$, three equal to $+1$, and sixteen equal to $+2$.

\begin{table*}[!t]
\centering
\caption{The $54$ nonzero coefficients of the expanded entropy functional.}
\label{tab:coefficients}
\scriptsize
\setlength{\tabcolsep}{3.0pt}
\renewcommand{\arraystretch}{1.03}
\begin{tabular}{@{}r l@{\quad}r l@{\quad}r l@{}}
\hline
$\kappa_S$ & scope & $\kappa_S$ & scope & $\kappa_S$ & scope\\
\hline
-1 & $H(a_{23})$ & -1 & $H(x_{2},y_{2})$ & +2 & $H(x_{1},y_{2},q_{12})$\\
-1 & $H(c)$ & -1 & $H(x_{2},y_{3})$ & +2 & $H(x_{1},y_{3},q_{13})$\\
-1 & $H(d)$ & -1 & $H(x_{2},y_{4})$ & -1 & $H(x_{1},y_{4},d)$\\
-1 & $H(q_{12})$ & -1 & $H(x_{3},y_{1})$ & +2 & $H(x_{2},y_{1},q_{12})$\\
-1 & $H(q_{13})$ & -1 & $H(x_{3},y_{2})$ & +2 & $H(x_{2},y_{2},a_{23})$\\
-1 & $H(q_{23})$ & -1 & $H(x_{3},y_{3})$ & +2 & $H(x_{2},y_{3},q_{23})$\\
-1 & $H(q_{24})$ & -1 & $H(x_{3},y_{4})$ & +2 & $H(x_{2},y_{4},q_{24})$\\
-1 & $H(q_{34})$ & -1 & $H(x_{4},y_{1})$ & +2 & $H(x_{3},y_{1},q_{13})$\\
-1 & $H(q_{12},q_{34})$ & -1 & $H(x_{4},y_{2})$ & +2 & $H(x_{3},y_{2},q_{23})$\\
-1 & $H(q_{13},q_{24})$ & -1 & $H(x_{4},y_{3})$ & +2 & $H(x_{3},y_{3},a_{23})$\\
-1 & $H(q_{23},c)$ & -1 & $H(x_{4},y_{4})$ & +2 & $H(x_{3},y_{4},q_{34})$\\
-1 & $H(x_{1},y_{2})$ & -1 & $H(y_{1},a_{23})$ & +2 & $H(x_{4},y_{1},d)$\\
-1 & $H(x_{1},y_{3})$ & +1 & $H(y_{4},d)$ & +2 & $H(x_{4},y_{2},q_{24})$\\
-1 & $H(x_{2},y_{1})$ & +2 & $H(q_{12},q_{34},c)$ & +2 & $H(x_{4},y_{3},q_{34})$\\
+2 & $H(q_{13},q_{24},c)$ & +2 & $H(q_{23},c,d)$ & +1 & $H(x_{4},y_{4},a_{23})$\\
+1 & $H(x_{1},y_{1},a_{23})$ & -1 & $H(q_{12},q_{34},q_{13},q_{24},c)$ & -1 & $H(x_{1},x_{2},y_{1},y_{2},q_{12})$\\
-1 & $H(x_{1},x_{3},y_{1},y_{3},q_{13})$ & -1 & $H(x_{2},x_{3},y_{2},y_{3},a_{23})$ & -1 & $H(x_{2},x_{3},y_{2},y_{3},q_{23})$\\
-1 & $H(x_{2},x_{4},y_{2},y_{4},q_{24})$ & -1 & $H(x_{3},x_{4},y_{3},y_{4},q_{34})$ & -1 & $H(x_{4},y_{1},q_{23},c,d)$\\
\hline
\end{tabular}
\end{table*}

\section*{Supplementary Information}
A self-contained exact-verification package is archived at Zenodo,
\href{https://doi.org/10.5281/zenodo.21872692}{doi:10.5281/zenodo.21872692}.
It independently reconstructs the finite group and the sixteen normal subgroups,
checks the finite identities in Appendix~\ref{app:finite-certificate}, reproduces
the coefficients of Appendix~\ref{app:coefficients}, and evaluates the witness
exactly.  No mathematical claim in the paper depends on the archive.

\section*{Acknowledgment}
OpenAI ChatGPT was used throughout the manuscript for language and structural
editing and LaTeX assistance, and was also used for literature-search assistance
and for development and debugging of the optional verification code.  The author
independently checked the mathematical statements, references, and exact
certificates and takes full responsibility for the content.

\end{document}